\documentclass[english,runningheads]{llncs}

\usepackage[T1]{fontenc}
\usepackage[utf8]{inputenc}
\usepackage{array}
\usepackage{setspace}
\usepackage[format=hang]{caption}
\usepackage{amsmath}
\usepackage{amssymb}
\usepackage{graphicx}
\usepackage{float}
\usepackage{xspace}
\usepackage{booktabs,bold-extra}
\usepackage[inline]{enumitem}
\usepackage[title]{appendix}
\usepackage{ifthen}
\usepackage{ifmtarg}
\usepackage[nolist]{acronym}
\usepackage{stmaryrd} 
\usepackage{csquotes}
\usepackage[english]{babel}
\usepackage{listings}
\usepackage[dvipsnames]{xcolor}
\usepackage{adjustbox}
\usepackage[lighttt]{lmodern}
\usepackage{mathtools}

\usepackage[textsize=scriptsize]{todonotes}
\usepackage{changes}

\let\llncssubparagraph\subparagraph
\let\subparagraph\paragraph
\usepackage[compact]{titlesec}
\let\subparagraph\llncssubparagraph

\usepackage[
  linktocpage,%
  pdfborderstyle={/S/S/W 0},%
]{hyperref}

\usepackage{keybook}
\usepackage{keyjavadl}
\usepackage{keylistings}

\newcommand{\javaInl}[2][]{\text{\lstinline[mathescape,#1]!#2!}}
\newcommand{\openLoopScope}    [1][\pv{x}]{\openLoopScopeWithCnt[#1]{}}
\newcommand{\closeLoopScope}   [1][\pv{x}]{\closeLoopScopeWithCnt[#1]{}}

\newcommand\jParOp{\javaInl{(}\xspace}

\newcommand{\openLoopScopeWithCnt} [2][\pv{x}]{\ensuremath{\prescript{}{}  {\circlearrowright}^{#2}}_{#1}   \xspace}
\newcommand{\closeLoopScopeWithCnt}[2][\pv{x}]{\ensuremath{\prescript{#2}{#1}{\hspace{-1pt}\circlearrowleft}}\xspace}
\newcommand{\reducedstrut}{\vrule width 0pt height .9\ht\strutbox depth .9\dp\strutbox\relax}
\newcommand{\highlights}[1]{%
  \begingroup
  \setlength{\fboxsep}{0pt}%
  \colorbox{gray!50}{\reducedstrut$#1$\/}%
  \endgroup
}

\newcommand{\att}[1][]{\ensuremath{\javaInl{attempt}_{#1}}\xspace}
\newcommand{\cont}[0]{\javaInl{continuation}\xspace}
\newcommand{\attCont}[0]{\att-\cont}
\newcommand{\combinedAtt}[0]{\att[\java{l}^?]}

\newcommand{\unlabeledAttempt}[3][]{\ensuremath{\att[#1] \javaInl{\ \{\ } #2 \javaInl{\ \}\ } \cont \javaInl{\ \{\ } #3 \javaInl{\ \}}}}
\newcommand{\attempt}[3][\java{l}]{\unlabeledAttempt[#1]{#2}{#3}}
\newcommand{\combinedAttempt}[3][\java{l}]{\attempt[#1^{?}]{#2}{#3}}

\newcommand{\labels}[0]{\ensuremath{\mathcal{L}}\xspace}
\newcommand{\normal}[0]{\ensuremath{\mathcal{N}}\xspace}
\newcommand{\breaks}[1]{\ensuremath{\mathcal{B}_{#1}}\xspace}
\newcommand{\continues}[1]{\ensuremath{\mathcal{C}_{#1}}\xspace}
\newcommand{\abrupt}[0]{\ensuremath{\mathcal{A}}\xspace}
\newcommand{\completionTypes}[0]{\ensuremath{\mathcal{T}}\xspace}

\newcommand{\halt}{\ensuremath{\downarrow}\xspace}

\newcommand{\ftrue}[0]{\mathit{true}}

\title{Treating \java{for}-Loops as\\First-Class Citizens in Proofs%
\thanks{\scriptsize This work was funded by the Hessian LOEWE initiative within the Software-Factory 4.0 project.}}
\author{Nathan Wasser \and Dominic Steinh\"ofel\orcidID{0000-0003-4439-7129}}
\institute{Technische Universit\"at Darmstadt, Department of Computer Science,\\
  64289 Darmstadt, Germany\\
\email{$\{$wasser,steinhoefel$\}$@cs.tu-darmstadt.de}\\
\url{https://www.informatik.tu-darmstadt.de/se}}
\date{\today}

\begin{document}

\maketitle

\begin{abstract}
Indexed loop scopes have been shown to be a helpful tool in creating sound loop invariant rules in dynamic logic for programming languages with abrupt completion, such as Java.
These rules do not require program transformation of the loop body, as other approaches to dealing with abrupt completion do.
However, indexed loop scopes were designed specifically to provide a loop invariant rule for \javaInl{while} loops and work rather opaquely. 
Here we propose replacing indexed loop scopes with a more transparent solution, which also lets us extend this idea from \javaInl{while} loops to \javaInl{for} loops. We further present sound loop unrolling rules for \javaInl{while}, \javaInl{do} and \javaInl{for} loops, which require neither program transformation of the loop body, nor the use of nested modalities.
This approach allows \javaInl{for} loops to be treated as first-class citizens in proofs -- rather than the usual approach of transforming \javaInl{for} loops into \javaInl{while} loops -- which makes semi-automated proofs more transparent and easier to follow for the user, whose interactions may be required in order to close the proofs.

\end{abstract}

\section{Introduction}

Sound program transformation in real world programming languages such as Java~\cite{JLS} is not easy, with potential pitfalls hiding in constructs such as Java's \javaInl{try}-\javaInl{finally} statement.
Thus, when reasoning about programs it is useful to avoid complex program transformations whenever possible.

\emph{Indexed loop scopes} were introduced in~\cite{Wasser16} to allow a sound loop invariant rule (which does not require program transformation of the loop body) in dynamic logic~\cite{Harel00} for \javaInl{while} loops containing statements which complete abruptly~\cite[Chapter~14.1]{JLS}. In~\cite{SteinhofelW17} it was shown that an implementation of this new loop invariant rule in \KeY\footnote{\scriptsize\url{https://www.key-project.org/}}~\cite{KeYBook2016} also decreases proof size when compared to the existing~rule.

However, indexed loop scopes were tailored specifically to treat the case of applying a loop invariant to a \javaInl{while} loop. While we made attempts to re-use indexed loop scopes for loop unrolling~\cite{Wasser16} and application to \javaInl{for} loops~\cite{WasserS19}, these were suboptimal.

In this paper we refine the concept of the loop scope, splitting it into two distinct parts: the \javaInl{attempt}-\javaInl{continuation} statement providing a non-active prefix~\cite{KeYBook2016} for loop bodies; and the logic to determine whether the loop invariant or the original formula should be proven, which was rather opaquely contained in symbolic execution rules for loop scopes.
Splitting these orthogonal concerns allows using an \javaInl{attempt}-\javaInl{continuation} statement in simple loop unrolling rules for \javaInl{while}, \javaInl{do} and \javaInl{for} loops, which avoid program transformation of the loop body and do not require the use of nested modalities, as the approach in~\cite{Wasser16} did for \javaInl{while} loop unrolling.
It also allows for a more transparent loop invariant rule for \javaInl{while} loops and we can introduce a transparent loop invariant rule for \javaInl{for} loops, which also both avoid program transformation of the loop~body.

With this, we can treat \javaInl{for} loops fully as first-class citizens in proofs, without the need to transform them into \javaInl{while} loops, which involves non-trivial program transformation.

Section~\ref{sec:background} provides background on dynamic logic and JavaDL in particular, as well as on indexed loop scopes and the loop invariant rule using them.
In Section~\ref{sec:loop-unrolling} we introduce the \javaInl{attempt}-\javaInl{continuation} block and new specialized loop unrolling rules for each loop type.
We propose new specialized loop invariant rules for \javaInl{while} and \javaInl{for} loops in Section~\ref{sec:loop-inv-for},
while Section~\ref{sec:evaluation} contains an evaluation of previous work and the changes proposed in this paper.
In Section~\ref{sec:related-work} we compare this approach with related work.
Finally, we conclude and offer ideas for future work in Section~\ref{sec:conclusion}.

\section{Background}\label{sec:background}

One approach to \emph{deductive software verification}~\cite{Filliatre11} which has been quite useful is \emph{dynamic logic}~\cite{Harel00}.
The idea behind dynamic logic is to contain the program under test within the logic itself by use of dynamic logic \emph{modalities}.
Classically, for all formulae $\phi$ and all programs \java{p} the formula $\dlboxf{p}{\phi}$ holds iff $\phi$ holds in all terminating states reachable by executing \java{p}.
The dual is defined as: $\dlf{p}{\phi} \equiv \neg (\dlboxf{p}{(\neg \phi)})$.
Initially proposed using \emph{Kleene's regular expression operators}
as programming language, it has been extended to various other programming languages, in particular to Java~\cite{JLS} in \emph{Java dynamic logic} (JavaDL)~\cite{BeckertHS07}.
While Kleene's regular expression operators contain complexities such as non-determinism, which makes reasoning about them far from simple, there is no concept of \emph{abrupt completion}\footnote{In Java, statements can complete abruptly due to \javaInl{break}s, \javaInl{continue}s and \javaInl{return}s, while both statements and expressions can complete abruptly due to thrown exceptions~\cite[Chapter~14.1]{JLS}.}: either an operation completes normally or blocks.
Additionally, program elements in Java can ``catch'' these abrupt completions and execute different code due to them, then either complete normally or complete abruptly for the same or a different reason.
Thus, it is not as simple a matter to give meaning to $\dlbox{\javaInl{while (e) st}}{\phi}$ for the \javaInl{while} loop of a Java program,
while the axiom for \javaInl{while} in a simple WHILE language can be expressed through \emph{loop unrolling}:
\begin{align*}
 \dlboxf{WHILE (e) st}{\phi}~\equiv~\dlboxf{IF e \{ st WHILE (e) st \}}{\phi}
\end{align*}
One solution, proposed for example in~\cite{Schlager07}, would be to introduce new modalities for each type of completion.

\begin{definition}[Set of all labels, sets of completion types]
 \labels is an infinite set of labels.
 The set of completion types \completionTypes and its subsets \normal~(\emph{normal}), \abrupt~(\emph{abrupt}), \breaks{l}~(\emph{breaking}) and \continues{l}~(\emph{continuing completion types}) are given as:
 \begin{align*}
  \normal = \{\mathit{normal}\},~ \abrupt = \{\mathit{break},\mathit{continue}\} \cup \bigcup_{l \in \labels}\{\mathit{break}_{l},\mathit{continue}_{l}\},~   \completionTypes = \normal \cup \abrupt, \\
  \forall l \in \labels.~\breaks{l} = \{\mathit{break},\mathit{break}_l\}, \qquad 
  \forall l \in \labels.~\continues{l} = \{\mathit{normal},\mathit{continue},\mathit{continue}_l\} \ 
 \end{align*}
\end{definition}
We write $\dlboxf{p}{_{S}~\phi}$ as short form for $\bigwedge_{t \in S} (\dlboxf{p}{_{t}~\phi)}$.
The axioms given in this paper hold for all $l \in \labels$.
Fig.~\ref{fig:axioms-simple} contains relatively straightforward axioms for some simple Java statements, as well as the \javaInl{try}-\javaInl{finally} statement.
\begin{figure}
\begin{minipage}{.45\textwidth}
\small
\begin{align}
 \dlboxf{;}{_\normal~\phi} &~\equiv~ \phi \label{eq:axiom-skip-normal} \\
 \dlbox{\javaInl{break;}}{_{\mathit{break}}~\phi} &~\equiv~ \phi \\
 \dlbox{\javaInl{continue;}}{_{\mathit{continue}}~\phi} &~\equiv~ \phi \\
 \dlbox{\javaInl{break\ }l\java{;}}{_{\mathit{break}_l}~\phi} &~\equiv~ \phi \\
 \dlbox{\javaInl{continue\ }l\java{;}}{_{\mathit{continue}_l}~\phi} &~\equiv~ \phi
\end{align}
\end{minipage}
\hfill\begin{minipage}{.45\textwidth}
\small
\begin{align}
 \dlboxf{;}{_\abrupt~\phi} \label{eq:axiom-skip-abrupt} \\
 \dlbox{\javaInl{break;}}{_{\completionTypes \setminus \{ \mathit{break} \}}~\phi} \\
 \dlbox{\javaInl{continue;}}{_{\completionTypes \setminus \{ \mathit{continue} \}}~\phi} \\
 \dlbox{\javaInl{break\ }l\java{;}}{_{\completionTypes \setminus \{ \mathit{break}_l \}}~\phi} \\
 \dlbox{\javaInl{continue\ }l\java{;}}{_{\completionTypes \setminus \{ \mathit{continue}_l \}}~\phi}
\end{align}
\end{minipage}
\small
\begin{align}
 \dlboxf{st1 st2}{_\normal~\phi} &~\equiv~ \dlboxf{st1}{_\normal\dlboxf{st2}{_\normal~\phi}} \label{eq:axiom-seq-normal} \\
 \forall a \in \abrupt.~\dlboxf{st1 st2}{_{a}~\phi} &~\equiv~ \dlboxf{st1}{_{a}~\phi}~\land~\dlboxf{st1}{_\normal\dlboxf{st2}{_{a}~\phi}} \label{eq:axiom-seq-abrupt} \\
 \forall t \in \completionTypes.~\dlbox{\javaInl{if (e) st1 else st2}}{_{t}~\phi} &~\equiv~ \dlboxf{b = e;}{_\normal((\java{b} \rightarrow \dlboxf{st1}{_{t}~\phi}) \land (\neg\java{b} \rightarrow \dlboxf{st2}{_{t}~\phi}))} \label{eq:axiom-if}
\end{align}%
\vspace{-0.5cm}
\begin{align}
  \dlbox{\javaInl{try \{ p \} finally \{ q \}}}_\normal~\phi 
 &~\equiv~
 \dlboxf{p}_\normal \dlboxf{q}_\normal~\phi  \\
 \forall a \in \abrupt.~ \dlbox{\javaInl{try \{ p \} finally \{ q \}}}_a~\phi 
 &~\equiv~ \dlboxf{p}_a \dlboxf{q}_\normal~\phi ~\land~ \dlboxf{p}_\completionTypes \dlboxf{q}_a~\phi 
\end{align}%
\caption{Axioms for \emph{skip}, \javaInl{break}s, \javaInl{continue}s, \emph{sequence}, \javaInl{if} and \javaInl{try}-\javaInl{finally}}%
\label{fig:axioms-simple}%
\end{figure}

\noindent
We write ``\javaInl{if (e) st}'' as short form for ``\javaInl{if (e) st else ;}''.

Fig.~\ref{fig:axioms-while} contains axioms for the \javaInl{while} statement.
The axiom~\eqref{eq:axiom-while-normal} expresses that the loop can: (1.) continue normally, or by a matching \javaInl{continue} statement; and (2.) be exited normally or by a matching \javaInl{break} statement. Axiom \eqref{eq:axiom-while-other-loop} expresses that the loop can complete abruptly by a labeled \javaInl{break} or \javaInl{continue} that does not match the loop label.
Axiom~\eqref{eq:axiom-while-same-loop} expresses that a \javaInl{while} loop can never complete abruptly due to a matching \javaInl{break} or \javaInl{continue} statement.
\begin{figure}
{\small
\begin{align}
\dlbox{l\javaInl{: while (e) st}}_{\normal}~\phi
  ~\equiv~ \dlboxf{b = e;}_{\normal}(&(\neg \java{b} \rightarrow \phi)~\land~ \nonumber \\
 &(~~\java{b} \rightarrow (\dlboxf{st}{_{\breaks{l}}~\phi} ~\land \nonumber \\
 &\qquad\quad\ \, \dlboxf{st}{_{\continues{l}} \dlbox{l\javaInl{: while (e) st}}{_{\normal}~\phi}}))) \label{eq:axiom-while-normal} \\
  \forall t \in \bigcup_{k \in \labels \setminus \{l\}}\{\mathit{break}_k, \mathit{continue}_k\}. \qquad\ \, \nonumber \\
  \dlbox{l\javaInl{: while (e) st}}_t~\phi ~\equiv~ \dlboxf{b = e;}_{\normal}(&\java{b} \rightarrow (\dlboxf{st}{_{t}~\phi}~\land~\dlboxf{st}_{\continues{l}}\dlbox{l\javaInl{: while (e) st}}_{t}~\phi)) \label{eq:axiom-while-other-loop} \\
  \dlbox{l\javaInl{: while (e) st}}&_{\{\mathit{break}, \mathit{continue}, \mathit{break}_l, \mathit{continue}_l\}}~\phi \label{eq:axiom-while-same-loop}
\end{align}%
}%
\caption{Axioms for \javaInl{while}}%
\label{fig:axioms-while}%
\end{figure}

\noindent
While this approach of adding many new modalities provides a sound theoretical grounding, 
a calculus directly using these axioms as rules
is problematic in practice
(in particular when using \emph{symbolic execution}~\cite{King76}), as it becomes quite complex very quickly. It should be pointed out that the axioms for the modalities covering exception throwing and returning from a method are more involved than the somewhat simpler modalities dealing with \javaInl{break}s and \javaInl{continue}s.
%
%
Additionally, modalities need to be analyzed multiple times, as can be seen by applying~\eqref{eq:axiom-while-normal} to $\dlbox{\javaInl{l: while (e) \{ st1 st2 \}}}{_{\mathit{normal}}~\phi}$ and then simplifying with~\eqref{eq:axiom-seq-normal} and \eqref{eq:axiom-seq-abrupt}, leading to three separate occurences of $\dlboxf{b = e;}{_{\mathit{normal}}\dlboxf{st1}{_{\mathit{normal}}(\cdot)}}$.
Using symbolic execution, this involves multiple symbolic executions of the exact same program fragment in the same state with the same context, which is a waste of resources.

For these and other reasons, the authors of JavaDL chose to instead keep track of the context \emph{within the program part} of the modality, rather than creating additional modality types.
To this end they defined \emph{legal program fragments}~\cite{KeYBook2016}, which may occur in the program part of a modality:

\begin{definition}
 Let Prg be a Java program. A \emph{legal program fragment} \java{p} is a sequence of Java statements, where there are local variables $\java{a}_1, \ldots , \java{a}_n$ of Java types $\java{T}_1, \ldots , \java{T}_n$ such that extending Prg with an additional class \java{C}
 yields again a legal program according to the rules of the Java language specification~\cite{JLS}, except that \java{p} may refer to fields, methods and classes that are not visible in \java{C}, and \java{p} may contain \emph{extended Java statements} in addition to normal Java statements; where the class \java{C} is declared:
 \begin{center}
 \begin{minipage}{.8\textwidth}
 \begin{lstlisting}[mathescape]
  public class C {
    public static void m(T$_1$ a$_1$, $\ldots$, T$_n$ a$_n$) { p }
  }
 \end{lstlisting}
 \end{minipage}
 \end{center}
\end{definition}
In~\cite{KeYBook2016} the only \emph{extended Java statement} allowed was the \emph{method-frame}, a way to track the context of within which method call (of which object or class) a program fragment was to be executed.
This allows for method calls within a program fragment to be replaced with method-frames containing their expanded method bodies.

\begin{definition}
 The set of all \emph{JavaDL formulae} is defined as the smallest set containing all:
 \begin{itemize}
  \item first-order formulae,
  \item $\dlboxf{p}{\phi}$, where \java{p} is a legal program fragment and $\phi$ is a JavaDL formula, and
  \item $\upl \mathcal{U} \upr \phi$, where $\phi$ is a JavaDL formula and $\mathcal{U}$ is an \emph{update}.
 \end{itemize}
\end{definition}
\begin{definition}
 An \emph{update} $\mathcal{U}$ expresses state changes. An \emph{elementary update} $\java{x} \upd t$ represents the states where the variable \java{x} is set to the value of the term $t$, while a \emph{parallel update} $\mathcal{U}_1~\|~\mathcal{U}_2$ expresses both updates simultaneously (with a last-wins to resolve conflicts).
 Updates can be applied to terms ($\upl \mathcal{U} \upr t$), formulae ($\upl \mathcal{U} \upr \phi$) and other updates ($\upl \mathcal{U}_1 \upr \mathcal{U}_2$), creating new terms, formulae and updates representing the changed state.
\end{definition}
A legal program fragment has the form ``\java{$\pi$ st $\omega$}'', where the \emph{non-active prefix} $\pi$ initially consisted only of an arbitrary sequence of opening braces ``\java{\{}'', labels, beginnings ``\javaInl{method-frame(}$\ldots$\javaInl{) \{}'' of method invocation statements, and beginnings ``\javaInl{try \{}'' of \javaInl{try}-(\javaInl{catch})-\javaInl{finally} statements; \java{st} is the \emph{active statement}; and $\omega$ is the rest of the program, in particular including closing braces corresponding to the opening braces in $\pi$.
Certain active statements can interact with the non-active prefix.

JavaDL uses a \emph{sequent calculus} in which rules consist of one conclusion and any number of premisses, and are applied bottom-up.
In addition to first-order logic rules, there are \emph{symbolic execution rules}, which operate on the active statement inside a legal program fragment. 
\begin{example}
 We consider:
 \(
   \upl \java{x} \upd 1 \upr\dlbox{\javaInl{l : \{ y = x; break l; y = 0; \}}}(y \neq 0)
 \)
 
 Here ``\java{l : \{}'' is the non-active prefix, while ``\java{y = x;}'' is the active statement.
 JavaDL contains a symbolic execution rule to execute a simple assignment, which leads to the formula
 $\upl \java{x} \upd 1 \upr \upl y \upd x \upr \dlbox{\javaInl{l : \{ break l; y = 0; \}}}(y \neq 0)$.
 Now the active statement ``\javaInl{break l;}'' interacts with the non-active prefix, removing the labeled block completely and leaving the formula  $\upl \java{x} \upd 1 \upr \upl y \upd x \upr \dlbox{}(y \neq 0)$ which is equivalent to $\upl \java{x} \upd 1 \upr \upl y \upd x \upr (y \neq 0)$. Applying the updates gives first $\upl \java{x} \upd 1 \upr (x \neq 0)$ and then $1 \neq 0$, which obviously holds.
 The update $\java{x} \upd 1$ could have alternatively been applied to the update $y \upd x$, yielding the parallel update $\java{x} \upd 1~\|~y \upd \upl \java{x} \upd 1 \upr x$, which simplifies to $\java{x} \upd 1~\|~y \upd 1$. Applying this update to $(y \neq 0)$ also leads to the formula $(1 \neq 0)$.
 \begin{center}
\begin{minipage}{.6\textwidth}
 \[\small
  \seqRule{assignment}%
 {\sequent{}{\upl \mathcal{U} \upr \upl \java{x} \upd se \upr \dlboxf{$\pi$ $\omega$}{\phi}}
 }%
 {\sequent{}{\upl \mathcal{U} \upr \dlboxf{$\pi$ x = $se$; $\omega$}{\phi}}}
 \]
\end{minipage}
\end{center}
 \begin{center}
\begin{minipage}{.8\textwidth}
 \[\small
  \seqRule{blockBreak}%
 {\sequent{}{\upl \mathcal{U} \upr \dlboxf{$\pi$ $\omega$}{\phi}}
 }%
 {\sequent{}{\upl \mathcal{U} \upr \dlboxf{$\pi$ $l_1, \ldots, \java{l},\ldots l_n$: \{ break l; p \} $\omega$}{\phi}}}
 \]
\end{minipage}
\end{center}
\end{example}
Initially there was no designated non-active prefix that allowed interaction with unlabeled \javaInl{break}s as well as labeled and unlabeled \javaInl{continue}s, which can occur in loop bodies.
This makes a simple loop unrolling rule impossible,
therefore loop bodies were transformed when unrolling the loop or when applying the loop invariant rule directly to a \javaInl{while} loop.
With a \javaInl{for} loop, the entire loop was first transformed, creating a \javaInl{while} loop, with a further program transformation of the loop body when dealing with said \javaInl{while} loop.
However, sound program
transformation rules for a complex language such as Java lead to very opaque program fragments, which have next to no relation to the original program, as can be seen in Examples~\ref{ex:for-to-while} and \ref{ex:while-transform}.
\begin{figure*}[b]
 \begin{minipage}{.45\textwidth}
 \begin{example}\label{ex:for-to-while}
 Sound program transformation of the \javaInl{for} loop in Listing~\ref{lst:ex_for} leads to the \javaInl{while} loop in Listing~\ref{lst:ex_while}.
 
 \begin{center}
 \begin{minipage}{.85\textwidth}
 \small
 \begin{lstlisting}[caption={Original \javaInl{for} loop},label=lst:ex_for]
for (; x > 1; x = x / 2) {
  if (x % 2 == 0) continue;
  if (x % 5 == 0) break;
}
\end{lstlisting}
\end{minipage}
\end{center}
\end{example}
 \end{minipage}\hfill
 \begin{minipage}{.46\textwidth}\small
\begin{lstlisting}[caption={Transformed \javaInl{while} loop},label=lst:ex_while]
b: {
     while (x > 1) {
c:     {
         if (x % 2 == 0) break c;
         if (x % 5 == 0) break b;
       }
       x = x / 2;
     }
   }
\end{lstlisting}
 \end{minipage}
 \end{figure*}
 
 \begin{figure*}[tbp]
\begin{minipage}{.45\textwidth}
\begin{example}\label{ex:while-transform}
 Consider the program fragment in Listing~\ref{lst:ex_abrupt}.
 Sound program transformation of this loop's body (in order to apply the loop invariant rule) must track abrupt completion within the body, but also reset and restore this tracking when encountering the \javaInl{finally} block to ensure that the semantics are not altered. This leads to the program fragment shown in Listing~\ref{lst:ex_abrupt_transformed_body}.

 \begin{center}
 \begin{minipage}{.7\textwidth}
 \small
\begin{lstlisting}[caption={Original loop},label=lst:ex_abrupt]
while (x != 0) {
  try {
    if (x > 0) return x;
    x = x + 100;
    break;
  } finally {
    if (x > 10) {
      x = -1;
      continue;
    }
  }
}
\end{lstlisting}
\end{minipage}
\end{center}
After executing this transformed loop body, the proof then continues on multiple branches for: (1.) the ``preserves invariant'' case where \java{brk} and \java{rtn} are \javaInl{false}, and \java{thrown} is \javaInl{null}; (2.) the ``exceptional use case'' where \java{thrown} is not null; (3.) the ``return use case'' where \java{rtn} is \javaInl{true}; and (4.) the ``break use case'' where \java{brk} is \javaInl{true}.
\end{example}
\end{minipage}
\hfill
\begin{minipage}{.45\textwidth}
\small
\begin{lstlisting}[caption={Transformed loop body},label=lst:ex_abrupt_transformed_body]
Throwable thrown = null;
boolean brk = false;
boolean cnt = false;
boolean rtn = false;
int rtnVal = 0;
try {
  l: {
    try {
      if (x > 0) {
        rtnVal = x;
        rtn = true;
        break l;
      }
      x = x + 100;
      brk = true;
      break l;
    } finally {
      boolean saveBrk = brk;
      brk = false;
      boolean saveCnt = cnt;
      cnt = false;
      boolean saveRtn = rtn;
      rtn = false;
      if (x > 10) {
        x = -1;
        cnt = true;
        break l;
      }
      brk = saveBrk;
      cnt = saveCnt;
      rtn = saveRtn;
    }
  }
} catch (Throwable t) {
  thrown = t;
}
\end{lstlisting}
\end{minipage}
\end{figure*}

\noindent
In~\cite{Wasser16} the concept of an \emph{indexed loop scope} (a further \emph{extended Java statement} $\openLoopScope~\java{st}~\closeLoopScope$) was proposed, allowing a designated non-active prefix for loop bodies (although the semantics of the indexed loop scope were such that it is directly useful only for a loop invariant rule for \javaInl{while} loops). Symbolic execution rules for \javaInl{continue}s and unlabeled \javaInl{break}s, as well as interaction between the various completion statements and the loop scope were defined. This allowed for the loop invariant rule below, which avoids program transformation of the loop body.
Additionally, it was shown in~\cite{SteinhofelW17} that an implementation of this rule in \KeY was more efficient than the loop invariant rule relying on program transformation.
\[\small
\seqRuleW{loopInvariantWhileWithLoopScopes}%
{\sequent{}{\upl \mathcal{U} \upr \mathit{Inv}} \\
  \sequentb{\mathit{Inv}}{\dlboxf{\(\pi \ \openLoopScope \javaInl{\ if (}\mathit{nse}\javaInl{) \{}\ \java{p}\ \javaInl{continue;} \ \java{\}}\ \closeLoopScope \ \omega\)}((\java{x} \doteq \text{FALSE} \rightarrow \mathit{Inv}){\quad} \\ \hfill \&\ (\java{x} \doteq \text{TRUE} \rightarrow \phi))}
}%
{\sequent{}{\upl \mathcal{U} \upr \dlbox{\pi \javaInl{\ while (}\mathit{nse}\javaInl{) p\ } \omega}{\phi}}}
\]
The first premiss ensures that the invariant holds in the program state before the first iteration of the loop.
The second premiss ensures both that normal and abrupt continuation of the loop body preserves the invariant; and that after leaving the loop normally or abruptly and executing the remaining program the original formula $\phi$ holds.
As this must hold for \emph{any} iteration, the assumptions $\Gamma \cup \neg\Delta$ and the update $\mathcal{U}$ expressing the program state before the first iteration are removed, with only the invariant as an assumption for the second premiss.

However, loop scopes work in a fairly opaque way: as can be seen in the rule above, the loop scope index \java{x} is never explicitly set anywhere in the rule, but rather will implicitly be set by the symbolic execution rules operating on loop scopes (with \javaInl{continue} setting it to false, and everything else setting it to true). In this paper we show how to create a more transparent solution.

\section{New Loop Unrolling Rules for JavaDL}\label{sec:loop-unrolling}

In order to introduce new loop unrolling rules specifically for \javaInl{while}, \javaInl{do} and \javaInl{for} loop, which do not require program transformation of the loop bodies, we require a non-active prefix for loop bodies in JavaDL. To this end we introduce the \attCont statement:

\subsection{Introducing the \attCont Statement}

\begin{definition}
  An \attCont statement is an \emph{extended Java statement} of the
 form ``\attempt{\javaInl{p}}{\javaInl{q}}''
 where $\java{l} \in \labels$ is a label, and \java{p} and \java{q} are (extended) Java statements.
 Non-active prefixes may additionally contain beginnings
 ``$\att[\javaInl{l}]\ \javaInl{\{}$'' of \attCont statements.
\end{definition}
If \java{p} does not contain any labeled \javaInl{break} or \javaInl{continue} statements matching the label \java{l}, ``\attempt{\javaInl{p}}{\javaInl{q}}'' is equivalent to its \emph{unlabeled} counterpart ``\unlabeledAttempt{\javaInl{p}}{\javaInl{q}}''.
Non-active prefixes may therefore contain unlabeled \attCont beginnings ``\att \javaInl{\{}''.

The semantic meaning of \attempt{\javaInl{p}}{\javaInl{q}} is that \java{p} is executed first, then there is a choice:
\begin{enumerate}
 \item If \java{p} completes normally or completes abruptly due to a matching \javaInl{continue} statement (\javaInl{continue l;} or \javaInl{continue;}), \java{q} is executed and the statement {\small\attempt{\javaInl{p}}{\javaInl{q}}} completes for the same reason as \java{q}.
 \item If \java{p} completes abruptly due to a matching \javaInl{break}  (\javaInl{break l;} or \javaInl{break;}), \java{q} is \emph{not} executed and {\small\attempt{\javaInl{p}}{\javaInl{q}}} completes normally.
 \item If \java{p} completes abruptly for any other reason (including due to a statement \javaInl{continue l';} or \javaInl{break l';} where $\java{l} \neq \java{l'}$), \java{q} is \emph{not} executed and {\small\attempt{\javaInl{p}}{\javaInl{q}}} completes abruptly for the same reason \java{p} completed abruptly.\label{abrupt-completion-of-attempt}
\end{enumerate}
%
%
Axioms for \attCont statements are shown in Fig.~\ref{fig:axioms-attempt-continuation}.
\begin{figure}
{\small
\begin{align}
 \dlboxf{\attempt{\javaInl{p}}{\javaInl{q}}}{_{\normal}~\phi} 
 &~\equiv~ \dlboxf{p}_{\continues{l}}\dlboxf{q}_{\normal}~\phi ~\land~ \dlboxf{p}_{\breaks{l}}~\phi \label{eq:axiom-attempt-normal} \\
 \forall t \in \{\mathit{break}, \mathit{continue}, \mathit{break}_l, \mathit{continue}_l\}. \qquad\qquad\, \nonumber \\
 \dlboxf{\attempt{\javaInl{p}}{\javaInl{q}}}{_{t}~\phi} \label{eq:axiom-attempt-same-loop}
 &~\equiv~ \dlboxf{p}_{\continues{l}}\dlboxf{q}_{t}~\phi \\
 \forall t \in \bigcup_{k \in \labels \setminus \{l\}}\{\mathit{break}_k, \mathit{continue}_k\}. \qquad\qquad\qquad\qquad \nonumber \\
 \dlboxf{\attempt{\javaInl{p}}{\javaInl{q}}}{_{t}~\phi} 
 &~\equiv~ \dlboxf{p}_{\continues{l}}\dlboxf{q}_{t}~\phi~\land~\dlboxf{p}_{t}~\phi \label{eq:axiom-attempt-other-loop}
\end{align}%
}%
\caption{Axioms for \attCont}%
\label{fig:axioms-attempt-continuation}%
\end{figure}

Correct unrolling of a \javaInl{while} loop is now possible with the help of \attCont statements, as shown in Theorem~\ref{theorem:correct-unrolling}.
\begin{theorem}[Correctness of loop unrolling]\label{theorem:correct-unrolling}
$\small\dlbox{\javaInl{l: while (e) st}}{_{t}~\phi}$ is equivalent to $\small\dlbox{\javaInl{if (e)\ }\attempt{\javaInl{st}}{\javaInl{l: while (e) st}}}{_{t}~\phi}$ for all completion types $t \in \completionTypes$.
\begin{proof}
 See appendix.
\end{proof}
\label{thm:correctness-loop-unrolling}
\end{theorem}
\subsection{Symbolic execution rules for \attCont}

We introduce new symbolic execution rules for the \attCont statement into JavaDL as follows:
\subsection*{For an empty \att block:}
\[\small
 \seqRule{emptyAttempt}%
 {\sequent{}{\upl \mathcal{U} \upr \dlboxf{$\pi$ q $\omega$}{\phi}}
 }%
 {\sequent{}{\upl \mathcal{U} \upr \dlboxf{$\pi$ \combinedAttempt{}{\javaInl{q}} $\omega$}{\phi}}}
\]
We combine two rules into one here, by writing ``\combinedAtt'' to express that there is a rule for the labeled \attCont statement and a rule for the unlabeled \attCont statement.
\subsection*{For an \att block with a leading \javaInl{continue} statement:}
\[\small
 \seqRuleW{attemptContinueNoLabel}%
 {\sequent{}{\upl \mathcal{U} \upr \dlboxf{$\pi$ q $\omega$}{\phi}}
 }%
 {\sequent{}{\upl \mathcal{U} \upr \dlboxf{$\pi$ \combinedAttempt{\javaInl{continue; p}}{\javaInl{q}} $\omega$}{\phi}}}
\]
\[\small
 \seqRuleW{attemptContinue}%
 {\sequent{}{\upl \mathcal{U} \upr \dlboxf{$\pi$ q $\omega$}{\phi}}
 }%
 {\sequent{}{\upl \mathcal{U} \upr \dlboxf{$\pi$ \attempt{\javaInl{continue l; p}}{\javaInl{q}} $\omega$}{\phi}}}
\]
\[\small
 \text{$\java{l} \neq \java{l}'$:}
 \seqRuleW{attemptContinueNoMatch}%
 {\sequent{}{\upl \mathcal{U} \upr \dlbox{\pi\javaInl{\ continue l}$'$\java{; $\omega$}}{\phi}}
 }%
 {\sequent{}{\upl \mathcal{U} \upr \dlboxf{$\pi$ \combinedAttempt{\javaInl{continue l}$'$\javaInl{; p}}{\javaInl{q}} $\omega$}{\phi}}}
\]
\subsection*{For an \att block with a leading \javaInl{break} statement:}
\[\small
 \seqRuleW{attemptBreakNoLabel}%
 {\sequent{}{\upl \mathcal{U} \upr \dlboxf{$\pi$ $\omega$}{\phi}}
 }%
 {\sequent{}{\upl \mathcal{U} \upr \dlboxf{$\pi$ \combinedAttempt{\javaInl{break; p}}{\javaInl{q}} $\omega$}{\phi}}}
\]
\[\small
 \seqRuleW{attemptBreak}%
 {\sequent{}{\upl \mathcal{U} \upr \dlboxf{$\pi$ $\omega$}{\phi}}
 }%
 {\sequent{}{\upl \mathcal{U} \upr \dlboxf{$\pi$ \attempt{\javaInl{break l; p}}{\javaInl{q}} $\omega$}{\phi}}}
\]
\[\small
 \text{$\java{l} \neq \java{l}'$:}
 \seqRuleW{attemptBreakNoMatch}%
 {\sequent{}{\upl \mathcal{U} \upr \dlbox{\pi\javaInl{\ break l}$'$\java{; $\omega$}}{\phi}}
 }%
 {\sequent{}{\upl \mathcal{U} \upr \dlboxf{$\pi$ \combinedAttempt{\javaInl{break l}$'$\java{; p}}{\javaInl{q}} $\omega$}{\phi}}}
\]
\subsection*{For an \att block with a leading \javaInl{throw} statement:}
\[\small
 \seqRuleW{attemptThrow}%
 {\sequent{}{\upl \mathcal{U} \upr \dlbox{\pi\javaInl{\ throw\ }\mathit{se}\java{; $\omega$}}{\phi}}
 }%
 {\sequent{}{\upl \mathcal{U} \upr \dlboxf{$\pi$ \combinedAttempt{\javaInl{throw\ }\mathit{se}\java{; p}}{\javaInl{q}} $\omega$}{\phi}}}
\]
\subsection*{For an \att block with a leading \javaInl{return} statement:}
\[\small
 \seqRuleW{attemptEmptyReturn}%
 {\sequent{}{\upl \mathcal{U} \upr \dlbox{\pi\javaInl{\ return;\ }\omega}{\phi}}
 }%
 {\sequent{}{\upl \mathcal{U} \upr \dlboxf{$\pi$ \combinedAttempt{\javaInl{return; p}}{\javaInl{q}} $\omega$}{\phi}}}
\]
\[\small
 \seqRuleW{attemptReturn}%
 {\sequent{}{\upl \mathcal{U} \upr \dlbox{\pi\javaInl{\ return\ }\mathit{se}\java{; $\omega$}}{\phi}}
 }%
 {\sequent{}{\upl \mathcal{U} \upr \dlboxf{$\pi$ \combinedAttempt{\javaInl{return\ }\mathit{se}\java{; p}}{\javaInl{q}} $\omega$}{\phi}}}
\]
Further symbolic execution rules in JavaDL for \javaInl{continue} statements and unlabeled \javaInl{break} statements when encountering other non-active prefixes are identical to those given in~\cite{Wasser16}.
These merely propagate the abruptly completing statements upwards (executing the \javaInl{finally} block first, in the case of a \javaInl{try}-(\javaInl{catch})-\javaInl{finally} statement). As an example, where $\mathit{cs}$ is a possibly empty list of \javaInl{catch}-blocks:
\[\small
 \seqRule{tryContinueNoLabel}%
 {\sequent{}{\upl \mathcal{U} \upr \dlbox{\pi\javaInl{\ r continue;\ }\omega}{\phi}}
 }%
 {\sequent{}{\upl \mathcal{U} \upr \dlbox{\pi\javaInl{\ try \{ continue; p \}\ }\mathit{cs}\javaInl{\ finally \{ r \}\ }\omega}{\phi}}}
\]

\subsection{JavaDL Loop Unwinding Rules using \attCont}\label{subsec:loop-unrolling}

We can also use \attCont statements in a loop unwinding rule for \javaInl{while} loops in JavaDL. This does not require nested modalities as used in~\cite{Wasser16}:
\[\small
\seqRuleW{unwindWhileLoop}%
{\sequent{}{\upl \mathcal{U} \upr [\pi\javaInl{\ if (}\mathit{nse}\java{) }\combinedAtt\java{ \{ p \}} \\
\qquad\qquad\qquad\qquad\quad\ \ \javaInl{continuation \{ l}^?\javaInl{: while (}\mathit{nse}\java{) p \} $\omega$}]{\phi}}
}%
{\sequent{}{\upl \mathcal{U} \upr \dlbox{\java{$\pi$ l$^?$: }\javaInl{while (}\mathit{nse}\java{) p $\omega$}}{\phi}}}
\]
We unroll and execute one iteration of the loop, winding up back at the beginning of the loop unless the loop body completes abruptly (not due to a matching \javaInl{continue}).
This closely resembles the loop unrolling equivalence in Theorem~\ref{theorem:correct-unrolling}.

The loop unwinding rule for \javaInl{do} loops is almost the same, except that the condition is not checked before the first iteration:
\[\small
\seqRuleW{unwindDoLoop}%
{\sequent{}{\upl \mathcal{U} \upr \dlboxf{\(\pi\ \combinedAttempt{\java{p}}{\java{l}^?\javaInl{: while (}\mathit{nse}\java{) p}} \omega\)}\phi}
}%
{\sequent{}{\upl \mathcal{U} \upr \dlbox{\pi\javaInl{\ l}^?\javaInl{: do p while (}\mathit{nse}\java{); $\omega$}}{\phi}}}
\]
As can be seen, a single loop unwinding turns a \javaInl{do} loop into a \javaInl{while} loop.

We can also introduce a loop unwinding rule for the \javaInl{for} loop.
As will be seen later, we have a rule to pull out the initializer of the \javaInl{for} loop, so the rule only considers \javaInl{for} loops with empty initializers:
\[\small
\seqRuleW{unwindForLoop}%
{\sequent{}{\upl \mathcal{U} \upr [\pi\javaInl{\ if (}g'\java{) } \combinedAtt \java{ \{ p \}} \\
\qquad\qquad\qquad\qquad\ \ \ \javaInl{continuation \{\ }\mathit{upd}'\javaInl{\ l}^?\javaInl{: for (;\ }g\javaInl{;\ }\mathit{upd}\javaInl{) p \}\ }\omega]{\phi}}
}%
{\sequent{}{\upl \mathcal{U} \upr \dlbox{\pi\javaInl{\ l}^?\javaInl{: for (;\ }g\javaInl{;\ }\mathit{upd}\javaInl{) p\ }\omega}{\phi}}}
\]
Here $\mathit{upd}'$ is a statement list equivalent to the expression list $\mathit{upd}$, and $g'$ is an expression equivalent to the guard $g$ (\javaInl{true}, if $g$ is empty).

As in the rules for \javaInl{while} and \javaInl{do} loops, the loop body is executed in an \javaInl{attempt} block. But before re-entering the loop in the continuation, we execute the \javaInl{for} loop's update. This ensures that we execute the \javaInl{for} loop's update whether the loop body completes normally or completes abruptly due to a matching \javaInl{continue} statement.

\section{New Loop Invariant Rules for JavaDL}\label{sec:loop-inv-for}

In order for the loop invariant rule based on loop scopes to be sound, when a \java{continue} statement reached a loop scope the appropriate symbolic execution rule in JavaDL needed to opaquely do two things: (1.) set the loop scope index to false and (2.) remove the entire surrounding legal program fragment. 
Thanks to \attCont statements we can explcitly set a variable in the continuation in order to transparently solve the first of these issues. However in order to solve the second issue transparently, we require the addition of a further extended Java statement, which explicitly halts the program. 

\subsection{Introducing the Halt Statement}

\begin{definition}
 The \emph{halt statement} (written \halt) is an \emph{extended Java statement} that, when executed, immediately halts the entire legal program fragment in which it is contained, ensuring that no further statements are executed (not even statements in \javaInl{finally} blocks).  
\end{definition}
The dynamic logic with modalities for each type of completion can be extended with new modalities $\dlboxf{p}{_{\halt}(\cdot)}$ for all legal program fragments \java{p}.
Axioms for \halt and the new modalities are shown in Fig.~\ref{fig:axioms-halt}.
In particular, loop unrolling using \attCont statements is also valid in the halt modalities:

\begin{theorem}[Correctness of loop unrolling in the halt modalities]
The formulae 
$\small\dlbox{\javaInl{if (e)\ }\attempt{\javaInl{st}}{\javaInl{l: while (e) st}}}{_{\halt}~\phi}$
and
$\small\dlbox{\javaInl{l: while (e) st}}{_{\halt}~\phi}$ are equivalent.
\label{thm:correctness-loop-unrolling-halt}
\end{theorem}

\begin{proof}
 See appendix.
\end{proof}

\begin{figure}
{\small%
 \begin{minipage}{.15\textwidth}
 \small%
 \begin{align}
 \dlboxf{;}{_{\halt}~\phi} \label{eq:axiom-skip-halt}
 \end{align}%
 \end{minipage}
 \begin{minipage}{.26\textwidth}
 \small%
 \begin{align}
 \dlbox{\javaInl{break;}}{_{\halt}~\phi} \\
 \dlbox{\javaInl{break\ }l\java{;}}{_{\halt}~\phi}
 \end{align}%
 \end{minipage}
 \begin{minipage}{.31\textwidth}
 \small%
 \begin{align}
 \dlbox{\javaInl{continue;}}{_{\halt}~\phi} \\
 \dlbox{\javaInl{continue\ }l\java{;}}{_{\halt}~\phi}
 \end{align}%
 \end{minipage}
 \begin{minipage}{.26\textwidth}
 \small%
 \begin{align}
 \dlboxf{\halt}{_{\halt}~\phi} &~\equiv~ \phi \\
 \dlboxf{\halt}{_{\completionTypes}~\phi}
 \end{align}%
 \end{minipage}
 \begin{align}
 \dlboxf{st1 st2}{_{\halt}~\phi} ~\equiv~& \dlboxf{st1}{_{\halt}~\phi}~\land~\dlboxf{st1}{_{\normal}\dlboxf{st2}{_{\halt}~\phi}}  
 \end{align}
 \vspace{-0.7cm}
 \begin{align}
&\dlbox{\javaInl{if (e) st1 else st2}}{_{\halt}\phi} \nonumber \\
~\equiv~& \dlboxf{b = e;}{_{\halt}\phi} ~\land~ \dlboxf{b = e;}_{\normal}((\java{b} \rightarrow \dlboxf{st1}{_{\halt}\phi})~\land~(\neg\java{b} \rightarrow \dlboxf{st2}{_{\halt}\phi})) \label{eq:axiom-if-halt} \\
&\dlbox{l\javaInl{: while (e) st}}{_{\halt}\phi} \nonumber \\
 ~\equiv~& \dlboxf{b = e;}_{\halt}\phi ~\land~ \dlboxf{b = e;}_{\normal}(\java{b} \rightarrow (\dlboxf{st}_{\halt}\phi~\land~
 \dlboxf{st}_{\continues{l}}\dlbox{l\javaInl{: while (e) st}}_{\halt}~\phi)) \label{eq:axiom-while-halt}
 \end{align}
 \vspace{-0.5cm}
 \begin{align}
 \dlbox{\javaInl{try \{ p \} finally \{ q \}}}_\halt \phi &~\equiv~
 \dlboxf{p}_\halt \phi ~\land~ \dlboxf{p}_\completionTypes \dlboxf{q}_\halt \phi  \\
  \dlboxf{\attempt{\javaInl{p}}{\javaInl{q}}}{_{\halt}\phi}
 &~\equiv~ \dlboxf{p}{_{\halt}\phi} ~\land~ \dlboxf{p}{_{\continues{l}}\dlboxf{q}{_{\halt}\phi}} \label{eq:axiom-attempt-halt}
\end{align}%
}%
\caption{Axioms for the halt statement and halt modality}%
\label{fig:axioms-halt}%
\end{figure}

\subsection*{Halting in JavaDL}

The single symbolic execution rule in JavaDL required for the halt statement is:
\[\small
 \seqRule{halt}%
 {\sequent{}{\upl \mathcal{U} \upr \phi}
 }%
 {\sequent{}{\upl \mathcal{U} \upr \dlboxf{$\pi$ \halt $\omega$}{\phi}}}
\]
Provided correct modalities for \javaInl{throw} and \javaInl{return}, as well as further axioms for missing Java statements (in particular \javaInl{throw}, \javaInl{try}-\javaInl{catch}, \javaInl{return}, \emph{assignment} and dealing with method calls),
Conjecture~\ref{conj:equiv-dl} claims equivalence between JavaDL and the dynamic logic with modalities for each type of completion.
\begin{conjecture}\label{conj:equiv-dl}
 The JavaDL formula $\dlboxf{p}{\phi}$ must hold iff $\phi$ holds in all normally completing or halting states reachable by executing \java{p}:
\begin{align*}
 \dlboxf{p}{\phi}~\equiv~\dlboxf{p}{_{\mathit{normal}}~\phi} ~\land~ \dlboxf{p}{_{\halt}\phi}
\end{align*}
\end{conjecture}

\subsection{Loop Invariant Rule for \javaInl{while} Loops using \attCont}

Thanks to \attCont and halt statements we introduce the following loop invariant rule
for \javaInl{while} loops, where \java{x} is a fresh boolean variable not occuring anywhere in the legal program fragment ``$\pi\javaInl{\ l}^?\javaInl{: while (}\mathit{nse}\javaInl{) p\ }\omega$'':
\[\small
\seqRuleW{loopInvariantWhile}%
{\sequent{}{\upl \mathcal{U} \upr \mathit{Inv}} \\
\sequentb{\mathit{Inv}}{[\pi\javaInl{\ x = true;} \\
  \qquad\qquad\javaInl{if (}\mathit{nse}\javaInl{)} \\
\qquad\qquad\qquad \combinedAttempt{p}{\javaInl{x = false;\ }\halt} \\
\qquad\ \ \, \omega]((\java{x} \doteq \text{FALSE} \rightarrow \mathit{Inv})\ \&\ (\java{x} \doteq \text{TRUE} \rightarrow \phi))}
}%
{\sequent{}{\upl \mathcal{U} \upr \dlbox{\pi \; \javaInl{l}^?\javaInl{: while (}\mathit{nse}\javaInl{) p} \;\omega}{\phi}}}
\]
As the \cont block is constructed only from a simple assignment and the halt statement, if \java{p} completes normally or completes abruptly due to a matching \java{continue}, it is guaranteed to set \java{x} to false and complete due to the halt statement, leaving the invariant to be proven in the state reached after execution of a single loop iteration.

In all other cases \java{x} retains its initial value true, leaving \(\upl \mathcal{U}' \upr \dlboxf{$\pi$ abrupt $\omega$}{\phi}\) to be proven, with $\mathcal{U}'$ expressing the state the program is in when the loop is left.
If $\mathit{nse}$ evaluates to false or \java{p} completes abruptly due to a matching \javaInl{break}, then \java{abrupt} is empty and it remains to prove $\upl \mathcal{U}' \upr \dlboxf{$\pi$ $\omega$}{\phi}$.
If \java{p} completes abruptly due to any other statement, \java{abrupt} is equal to that abruptly completing statement.

\subsection{Loop Invariant Rule for \javaInl{for} Loops using \attCont}

In order to prove that the loop invariant of a \javaInl{for} loop initially holds, we must first reach the ``initial'' entry point of the loop. This is the point after full execution of the loop initializer. We therefore introduce the following rule to pull out the loop initializer of a \javaInl{for} loop, where $init'$ is a statement list equivalent to the loop initializer $init$:
\[\small
\seqRule{pullOutLoopInitializer}%
{\sequent{}{\upl \mathcal{U} \upr \dlbox{\pi~\jParOp\mathit{init}'\java{ l}^?\javaInl{: for (;\ }\mathit{guard}\java{; }\mathit{upd}\javaInl{) p \}\ }\omega}{\phi}}
}%
{\sequent{}{\upl \mathcal{U} \upr \dlbox{\java{$\pi$ l$^?$: }\javaInl{for}\java{ ($init$; $guard$; $upd$) p $\omega$}}{\phi}}}
\]
The following loop invariant rule can then be applied to \javaInl{for} loops without loop initializers, where
\java{x} is a fresh boolean variable not occurring anywhere in the legal program fragment ``$\pi\java{ l}^?\javaInl{: for (;\ }\mathit{guard}\java{; }\mathit{upd}\javaInl{) p\ }\omega$'',
$\mathit{upd}'$ is a statement list equivalent to the expression list $\mathit{upd}$, and $\mathit{guard}'$ is an expression equivalent to the guard $\mathit{guard}$ (\javaInl{true}, if $\mathit{guard}$ is empty):
\[\small
\seqRuleW{loopInvariantFor}%
{\sequent{}{\upl \mathcal{U} \upr \mathit{Inv}} \\
\sequentb{\mathit{Inv}}{[\pi\javaInl{\ x = true;} \\
  \qquad\qquad \javaInl{if (}\mathit{guard}'\java{)} \\
  \qquad\qquad\qquad \combinedAttempt{p}{\mathit{upd}'\javaInl{\ x = false;\ }\halt} \\
\qquad\quad \omega]((\java{x} \doteq \text{FALSE} \rightarrow \mathit{Inv})\ \&\ (\java{x} \doteq \text{TRUE} \rightarrow \phi))}
}%
{\sequent{}{\upl \mathcal{U} \upr \dlbox{\pi~\java{l}^?\javaInl{: for (;\ }\mathit{guard}\java{; }\mathit{upd}\java{) p}~\omega}{\phi}}}
\]
As the continuation is constructed only from the modified \javaInl{for} loop update $\mathit{upd}'$, a simple assignment and the halt statement, it cannot contain \javaInl{break}s, \javaInl{continue}s or \javaInl{return}s. It also cannot contain an explicit \javaInl{throw}, but implicitly exceptions can be thrown in $\mathit{upd}'$. Thus if \java{p} completes normally or completes abruptly due to a matching \javaInl{continue}, causing symbolic execution of the continuation, this will either set \java{x} to false and complete due to the halt statement, leaving the invariant to be proven in the state reached after execution of a single loop iteration; or it will complete abruptly due to a statement \javaInl{throw se;} (keeping \java{x} set to its initial value of true), leaving \(\upl \mathcal{U}' \upr \dlbox{\pi\javaInl{\ throw se;\ }\omega}{\phi}\) to be proven, with $\mathcal{U}'$ expressing the state the program is in when the loop is left abruptly due to the exception.
All other cases are identical to those for the \javaInl{while} loop invariant above.

\begin{theorem}
  The symbolic execution loop invariant rules
  \ruleName{loopInvariantWhile} and \ruleName{loopInvariantFor}
  are sound.
  \label{thm:soundness-loop-inv-rule}
\end{theorem}

\begin{proof}[Sketch]
  Consider rule \ruleName{loopInvariantFor}. The sequent in the conclusion matches the corresponding one in rule \ruleName{unwindForLoop} which we assume to be sound (see, e.g., Thm.~\ref{thm:correctness-loop-unrolling}). We compare the active statements in the modalities of the premisses of those rules:
    \[
      \begin{array}{lr}
        \javaInl{if (}g'\java{) }\combinedAtt \java{ \{ p \}\ } 
        \javaInl{continuation \{} & \hspace{1cm}\text{(\ruleName{unwindForLoop})} \\
          \qquad\mathit{upd}'\java{ }\highlights{\javaInl{l}^?\javaInl{: for (;\ }g\javaInl{;\ }\mathit{upd}\javaInl{) p}} & \\
        \javaInl{\}} & \\[5pt]
        \highlights{\javaInl{x = true;}} & \hspace{1cm}\text{(\ruleName{loopInvariantFor})} \\
        \javaInl{if (}g'\java{) }\combinedAtt \java{ \{ p \}\ } 
        \javaInl{continuation \{} & \\
          \qquad\mathit{upd}'\java{ }\highlights{\javaInl{x = false;\ }\halt} & \\
        \javaInl{\}} &
      \end{array}
    \]
    The differences between these programs, highlighted in gray, are rather small. Apart from the additional program variable \java{x} used in \ruleName{loopInvariantFor}, the original loop in the \javaInl{continuation} part is replaced by a \emph{halt} statement. The rule \ruleName{loopInvariantFor} therefore ``prunes'' remaining iterations: Where in \ruleName{unwindForLoop} we would continue with more unwinding iterations, we remember that we normally would do this by setting the flag \java{x} to \javaInl{false} and thus prove the invariant formula. For all other cases, we know that the loop is not continued, therefore \java{x} remains \javaInl{true} and we continue with executing the remaining program. Since in the invariant rule, the leading update application $\upl\U\upr$ is removed, we prove the invariant for \emph{an arbitrary iteration}; together with the first premiss asserting that the invariant holds initially, $\sequent{}{\upl\U\upr\mathit{Inv}}$, this forms an \emph{inductive argument} which allows us to abstract the loop by the invariant in the proof cases where \java{x} remains \javaInl{true} and we continue symbolic execution. The argument for \ruleName{loopInvariantWhile} is similar, but simpler.
    \qed
\end{proof}

As can be seen, introducing \javaInl{attempt}-\javaInl{continuation} and halt statements has allowed us to have a loop invariant rule specifically for \javaInl{for} loops, which does not require program transformation of the loop body and only minimal program transformation of the loop update. This allows \javaInl{for} loops to be treated as first-class citizens in proofs and lets user interactions occur on legal program fragments which are still reasonably close to the original program, rather than on those which have been transformed in such a way that it is unclear how they relate to the original program. This increases the transparency of the proof.


\subsection{Why No Loop Invariant Rule for \javaInl{do} Loops?}\label{subsec:rules-do}

One could imagine that a similar case could be made to treat \javaInl{do} loops as first-class citizens in proofs, by supplying a loop invariant rule specifically for \javaInl{do} loops. However, this is not really the case.
As with the other loop types, the loop invariant for a \javaInl{do} loop needs to hold only just before the condition is checked.
However, unlike the other loop types, this is not the case for \javaInl{do} loops until after the first loop iteration.
This makes a loop invariant rule for \javaInl{do} loops actually \emph{less} transparent, than the reasonably simple steps of (1.) converting the \javaInl{do} loop into a \javaInl{while} loop and (2.) applying the loop invariant rule for \javaInl{while} loops on the resulting \javaInl{while} loop. This transformation of a \javaInl{do} loop into a \javaInl{while} loop can happen in one of two ways: (i) by applying the \ruleName{unwindDoLoop} rule to the \javaInl{do} loop and symbolically executing the unrolled body until the \javaInl{attempt}-block is exited and the \javaInl{while} loop in the \javaInl{continuation}-block becomes the active statement, or (ii) by applying the program transformation rule from~\cite{deGouw2019} to the \javaInl{do} loop, producing a \javaInl{while} loop directly without needing to symbolically execute the first loop iteration:
\[\small
\seqRuleW{transformDoToWhile}%
{\sequent{}{\upl \mathcal{U}~||~\java{fst}\upd\text{TRUE} \upr 
\dlbox{\java{l$^?$: }\javaInl{while (fst ||\ }\mathit{nse}\javaInl{) \{ fst = false; p \}\ }\omega}{\phi}}
}%
{\sequent{}{\upl \mathcal{U} \upr \dlbox{\java{$\pi$ l$^?$: }\javaInl{do p while (}\mathit{nse}\java{); }\omega}{\phi}}}
\]
Here \java{fst} is a fresh boolean variable. The loop invariant then applied to the resulting \javaInl{while} loop can use the value of \java{fst} if the invariant of the original \javaInl{do} loop is only established after at least one iteration of the loop has been executed. 

\section{Evaluation}\label{sec:evaluation}

Based on our previous work on providing a loop invariant rule specifically for \javaInl{for} loops using loop scopes~\cite{WasserS19}, Benedikt Dreher implemented this loop invariant in \KeY and evaluated it in~\cite{Dreher19}.
He found that the efficiency of the new rule was similar to the pure program transformation rule and the rule using loop scopes on \javaInl{while} loops produced by program transformation of the \javaInl{for} loop.
The new rule required only about 80\% as many nodes and execution steps as the pure program transformation rule, while creating slightly more branches (creating an average of $27.86$ to $27.5$ branches in the examples).
It was about 10\% less efficient than the rule using loop scopes on \javaInl{while} loops produced by program transformation of the \javaInl{for} loop.
However, the new rule provided more transparency, as it was easier to see in the proof tree which statement in the original \javaInl{for} loop was being processed, as well as seeing directly what the result of applying the loop invariant rule to a \javaInl{for} loop would produce.

The rules proposed in this paper should be slightly more efficient, as they do not require the unnecessary steps of resetting the loop scope index before symbolically executing the \javaInl{for} loop's update and then setting the loop scope index afterwards, as the implemented rule from~\cite{WasserS19} does.
Additionally, the transparency of the rules proposed in this paper should be even greater, as the opacity of the loop scope has been completely replaced with the transparency of \javaInl{attempt}-\javaInl{continuation} and halt statements.

\section{Related Work}\label{sec:related-work}

We have already compared our approach to other JavaDL approaches using program transformation of the loop body or indexed loop scopes, showing that our approach here is much more transparent.
We have also compared this approach to using a dynamic logic with typed modalities for each completion type, which has drawbacks in particular when using symbolic execution.
%
%
We unfortunately could not find any work formally explaining the handling of irregular control flow in loops for VeriFast~\cite{Jacobs11}, a symbolic execution system for C and Java; the most formal paper we could find~\cite{Jacobs15} describes only a reduced language without \javaInl{break}s and \javaInl{continue}s.
The symbolic execution calculus for KIV~\cite{Stenzel05} is also a dynamic logic variant.
However, they sequentially decompose (\emph{flatten}) statements, such that a non-active prefix is not needed.
This is accomplished by including both heavy program transformation and tracking of \emph{mode information}, which has similarities to using a dynamic logic with typed modalities for each completion type.
Additionally, their approach cannot deal directly with \javaInl{continue}s%
, as they claim that these are problematic for loop unwinding; we have shown that this is not the case with our approach, providing loop unwinding rules for not only \javaInl{while}, but also \javaInl{do} and \javaInl{for} loops.
OpenJML~\cite{Cok14} and other approaches using \emph{verification condition generation} work by translating the program into an intermediate language.
Abrupt completion is usually modelled by branches to basic blocks.
This might make these approaches efficient, but the treatment of \emph{all} loop types becomes completely opaque.
While intermediate languages are less complex (which can be helpful), the translation into them can require compromises concerning soundness~\cite{Flanagan01} and is a non-trivial and error-prone task~\cite{Marche04} in any case.

\section{Conclusion and Future Work}\label{sec:conclusion}

We have introduced \javaInl{attempt}-\javaInl{continuation} and halt statements as extended Java statements that allow more localized reasoning for loops and a way to express immediately halting the Java program.
Axioms for these statements and the appropriately typed modalities have been given in a dynamic logic with modalities for various completion types.
These statements are of particular interest in JavaDL, where we have supplied symbolic execution rules for them. 

We have shown that using \javaInl{attempt}-\javaInl{continuation} statements rather than indexed loop scopes lets us gain great potential:
\begin{enumerate}
 \item We are able to express a loop invariant rule specifically for \javaInl{for} loops which does not require program transformation of the loop body and allows a transparent treatment of \javaInl{for} loops as first-class citizens in proofs.
 \item We are able to express loop unrolling rules for \javaInl{while}, \javaInl{do} and \javaInl{for} loops which require neither program transformation of the loop body, nor the use of nested modalities.
 \item The rule for a \javaInl{continue} reaching the \javaInl{attempt}-block (the non-active prefix responsible for loop bodies) is more transparent than the corresponding rule for loop scopes, simply executing the continuation (whatever it may be), rather than opaquely setting the loop scope index to~false.
\end{enumerate}
As future work we will implement these ideas into \KeY, performing an evaluation of the loop invariant rules for \javaInl{while} and \javaInl{for} loops with this approach on the examples tested in~\cite{SteinhofelW17} and~\cite{Dreher19}, so as to compare them with the loop scope approach.
We would also like to evaluate the new loop unrolling rules and are looking to find an appropriate benchmark for that.

Additionally, we will look into adding a \emph{halts} clause to 
JML~\cite{JML} method contracts, in order to express what must hold if a method executes the halt statement.
While no Java method can syntactically contain the halt statement, the Java virtual machine does provide the effect of halting, with the methods \javaInl{Runtime.exit()} and \javaInl{System.exit()}~\cite[Chapter~12.8]{JLS}.
Providing a way to express halting in a method contract is therefore somewhat of interest. 

\section*{Acknowledgements}

We thank Benedikt Dreher for his implementation and evaluation of our previous attempt at using indexed loop scopes to create a loop invariant rule specifically for \javaInl{for} loops.
We thank Richard Bubel for the fruitful discussions leading to and during the writing of this paper.

\bibliographystyle{splncs04}
\bibliography{main}

\appendix

\section*{Appendix}

\subsection*{Proofs for the Theorems}

\setcounter{theorem}{0}

\begin{theorem}[Correctness of loop unrolling]
$\small\dlbox{\javaInl{l: while (e) st}}{_{t}~\phi}$ is equivalent to $\small\dlbox{\javaInl{if (e)\ }\attempt{\javaInl{st}}{\javaInl{l: while (e) st}}}{_{t}~\phi}$ for all completion types $t \in \completionTypes$.
\begin{proof}
$\small\dlbox{\javaInl{if (e)\ }\attempt{\javaInl{st}}{\javaInl{l: while (e) st}}}{_{t}~\phi}$ expands to: $\small\dlbox{\javaInl{if (e)\ }\attempt{\javaInl{st}}{\javaInl{l: while (e) st}} \javaInl{else ;}}_{t}~\phi$.

As $\abrupt = \{\mathit{break}, \mathit{continue}, \mathit{break}_{\java{l}}, \mathit{continue}_{\java{l}}\} \cup \bigcup_{k \in \labels \setminus \{\java{l}\}}\{\mathit{break}_k, \mathit{continue}_k\}$
and $\completionTypes = \normal \cup \abrupt$, by case distinction:
\begin{description}
 \item[If $t \in \normal$:]
 {\small
 \begin{align*}
  &\dlbox{\javaInl{if (e)\ }\attempt{\javaInl{st}}{\javaInl{l: while (e) st}}\javaInl{\ else ;}}_{\normal}~\phi \\
  \equiv~& \dlboxf{b = e;}_\normal \\ &\quad((\java{b} \rightarrow \dlboxf{\attempt{\javaInl{st}}{\javaInl{l: while (e) st}}}_{\normal}~\phi) \\
  &\quad~\land~(\neg\java{b} \rightarrow \dlboxf{;}_{\normal}~\phi)) \tag*{by \eqref{eq:axiom-if}} \\
  \equiv~& \dlboxf{b = e;}_\normal ((\java{b} \rightarrow (\dlboxf{st}_{\continues{l}}\dlbox{\javaInl{l: while (e) st}}_{\normal}~\phi ~\land~ \dlboxf{st}_{\breaks{l}}~\phi))~\land~(\neg\java{b} \rightarrow \dlboxf{;}_{\normal}~\phi)) \tag*{by \eqref{eq:axiom-attempt-normal}} \\
  \equiv~& \dlboxf{b = e;}_\normal ((\java{b} \rightarrow (\dlboxf{st}_{\continues{l}}\dlbox{\javaInl{l: while (e) st}}_{\normal}~\phi ~\land~ \dlboxf{st}_{\breaks{l}}~\phi))~\land~(\neg\java{b} \rightarrow \phi)) \tag*{by \eqref{eq:axiom-skip-normal}} \\
  \equiv~& \dlboxf{b = e;}_\normal ((\neg\java{b} \rightarrow \phi)~\land~(\java{b} \rightarrow(\dlboxf{st}_{\breaks{l}}~\phi~\land~\dlboxf{st}_{\continues{l}}\dlbox{\javaInl{l: while (e) st}}_{\normal}~\phi))) \tag*{by \emph{commutativity} of $\land$} \\
  \equiv~& \dlbox{\javaInl{l: while (e) st}}{_{\normal}~\phi} \tag*{by \eqref{eq:axiom-while-normal} \qed}
 \end{align*}%
 }%
 \item[If $t \in \{\mathit{break}, \mathit{continue}, \mathit{break}_{\java{l}}, \mathit{continue}_{\java{l}}\}$:]
 {\small
 \begin{align*}
  &\dlbox{\javaInl{if (e)\ }\attempt{\javaInl{st}}{\javaInl{l: while (e) st}}\javaInl{\ else ;}}_{t}~\phi \\
  \equiv~& \dlboxf{b = e;}_\normal \\ &\quad((\java{b} \rightarrow \dlboxf{\attempt{\javaInl{st}}{\javaInl{l: while (e) st}}}_{t}~\phi) \\
  &\quad~\land~(\neg\java{b} \rightarrow \dlboxf{;}_t~\phi)) \tag*{by \eqref{eq:axiom-if}} \\
  \equiv~& \dlboxf{b = e;}_\normal((\java{b} \rightarrow \dlboxf{st}_{\continues{l}}\dlbox{\javaInl{l: while (e) st}}_{t}~\phi)~\land~(\neg\java{b} \rightarrow \dlboxf{;}_t~\phi)) \tag*{by \eqref{eq:axiom-attempt-same-loop}} \\
  \equiv~& \dlboxf{b = e;}_\normal((\java{b} \rightarrow \dlboxf{st}_{\continues{l}}\dlbox{\javaInl{l: while (e) st}}_{t}~\phi)~\land~(\neg\java{b} \rightarrow \ftrue)) \tag*{by \eqref{eq:axiom-skip-abrupt}} \\
  \equiv~& \dlboxf{b = e;}_\normal(\java{b} \rightarrow \dlboxf{st}_{\continues{l}}\dlbox{\javaInl{l: while (e) st}}_{t}~\phi) \tag*{by definition of $\rightarrow$ and $\land$} \\
  \equiv~& \dlboxf{b = e;}_\normal(\java{b} \rightarrow \dlboxf{st}_{\continues{l}}~\ftrue) \tag*{by \eqref{eq:axiom-while-same-loop}} \\
  \equiv~& \dlboxf{b = e;}_\normal(\java{b} \rightarrow \ftrue) \tag*{by \emph{necessitation}} \\
  \equiv~& \dlboxf{b = e;}_\normal~\ftrue \tag*{by definition of $\rightarrow$} \\
  \equiv~& \ftrue \tag*{by \emph{necessitation}} \\
 \equiv~& \dlbox{\javaInl{l: while (e) st}}{_{t}~\phi} \tag*{by \eqref{eq:axiom-while-same-loop} \qed}
 \end{align*}%
 }%
 \item[Otherwise, $t \in \bigcup_{k \in \labels \setminus \{\java{l}\}}\{\mathit{break}_k, \mathit{continue}_k\}$:]
 {\small
 \begin{align*}
  &\dlbox{\javaInl{if (e)\ }\attempt{\javaInl{st}}{\javaInl{l: while (e) st}}\javaInl{\ else ;}}_{t}~\phi \\
  \equiv~& \dlboxf{b = e;}_\normal \\ &\quad((\java{b} \rightarrow \dlbox{\attempt{\javaInl{st}}{\javaInl{l: while (e) st}}}_{t}~\phi) \\
  &\quad~\land~(\neg\java{b} \rightarrow \dlboxf{;}_t~\phi)) \tag*{by \eqref{eq:axiom-if}} \\
  \equiv~& \dlboxf{b = e;}_\normal ((\java{b} \rightarrow (\dlboxf{st}_{\continues{l}}\dlbox{\javaInl{l: while (e) st}}_{t}~\phi~\land~\dlboxf{st}_{t}~\phi))~\land~(\neg\java{b} \rightarrow \dlboxf{;}_t~\phi))  \tag*{by \eqref{eq:axiom-attempt-other-loop}} \\
  \equiv~& \dlboxf{b = e;}_\normal ((\java{b} \rightarrow (\dlboxf{st}_{\continues{l}}\dlbox{\javaInl{l: while (e) st}}_{t}~\phi~\land~\dlboxf{st}_{t}~\phi))~\land~(\neg\java{b} \rightarrow \ftrue))  \tag*{by \eqref{eq:axiom-skip-abrupt}} \\
  \equiv~& \dlboxf{b = e;}_\normal (\java{b} \rightarrow (\dlboxf{st}_{\continues{l}}\dlbox{\javaInl{l: while (e) st}}_{t}~\phi~\land~\dlboxf{st}_{t}~\phi)) \tag*{by definition of $\rightarrow$ and $\land$} \\
  \equiv~& \dlboxf{b = e;}_\normal (\java{b} \rightarrow (\dlboxf{st}_{t}~\phi~\land~\dlboxf{st}_{\continues{l}}\dlbox{\javaInl{l: while (e) st}}_{t}~\phi)) \tag*{by \emph{commutativity} of $\land$} \\
 \equiv~& \dlbox{\javaInl{l: while (e) st}}{_{t}~\phi} \tag*{by \eqref{eq:axiom-while-other-loop} \qed}
 \end{align*}%
 }%
\end{description}
\end{proof}
\end{theorem}

\begin{theorem}[Correctness of loop unrolling in the halt modalities]
The formulae $\small\dlbox{\javaInl{if (e)\ }\attempt{\javaInl{st}}{\javaInl{l: while (e) st}}}{_{\halt}~\phi}$ and $\small\dlbox{\javaInl{l: while (e) st}}{_{\halt}~\phi}$ are equivalent.
\begin{proof}
 {\small
 \begin{align*}
  &\dlboxf{if (e) \attempt{\javaInl{st}}{\javaInl{l: while (e) st}}\javaInl{\ else ;}}{_{\halt}~\phi} \\
  ~\equiv~& \dlboxf{b = e;}{_{\halt}\phi} ~~\land~ \\
  &\dlboxf{b = e;}_{\normal}((\java{b} \rightarrow \dlboxf{\attempt{\javaInl{st}}{\javaInl{l: while (e) st}}}{_{\halt}\phi}) \\
  &\qquad\qquad\quad\land~(\neg\java{b} \rightarrow \dlboxf{;}_{\halt}\phi)) \tag*{by~\eqref{eq:axiom-if-halt}} \\
  ~\equiv~& \dlboxf{b = e;}{_{\halt}\phi} ~\land~ \dlboxf{b = e;}_{\normal}((\java{b} \rightarrow (\dlboxf{st}{_{\halt}\phi} ~\land~ \dlboxf{st}_{\continues{l}}\dlbox{\javaInl{l: while (e) st}}_{\halt}\phi)) \\
  &\qquad\qquad\qquad\qquad\qquad\quad~\land~(\neg\java{b} \rightarrow \dlboxf{;}_{\halt}\phi)) \tag*{by~\eqref{eq:axiom-attempt-halt}} \\
  ~\equiv~& \dlboxf{b = e;}{_{\halt}\phi} ~\land~ \dlboxf{b = e;}_{\normal}((\java{b} \rightarrow (\dlboxf{st}{_{\halt}\phi} ~\land~ \dlboxf{st}_{\continues{l}}\dlbox{\javaInl{l: while (e) st}}_{\halt}\phi)) \\
  &\qquad\qquad\qquad\qquad\qquad\quad~\land~(\neg\java{b} \rightarrow \ftrue)) \tag*{by~\eqref{eq:axiom-skip-halt}} \\
  ~\equiv~& \dlboxf{b = e;}{_{\halt}\phi} ~\land~ \dlboxf{b = e;}_{\normal}(\java{b} \rightarrow (\dlboxf{st}{_{\halt}\phi} ~\land~ \dlboxf{st}_{\continues{l}}\dlbox{\javaInl{l: while (e) st}}_{\halt}\phi)) \tag*{by definition of $\rightarrow$ and $\land$} \\
 ~\equiv~& \dlbox{\javaInl{l: while (e) st}}{_{\halt}~\phi} \tag*{by~\eqref{eq:axiom-while-halt} \qed}
 \end{align*}%
 }%
\end{proof}
\end{theorem}

\end{document}